\documentclass[runningheads,envcountsame,a4paper]{llncs}
\usepackage{amsmath,amssymb}
\usepackage{tikz}
\usetikzlibrary{chains,backgrounds}
\usetikzlibrary{arrows,automata}
\usepackage{url}
\usepackage{float}
\usepackage{array}
\usepackage{graphicx}
\usepackage{subfig}

\DeclareMathOperator{\cyc}{cyc}
\DeclareMathOperator{\dfa}{DFA}
\def\modd#1 #2{#1\ ({\rm mod}\ #2)}

\frontmatter

\title{Shortest Repetition-Free Words Accepted by Automata}

\author{Hamoon Mousavi \and Jeffrey Shallit}

\institute{School of Computer Science, University of Waterloo,
Waterloo, ON  N2L 3G1 Canada  \\
\email{ \{sh2mousa,shallit\}@uwaterloo.ca}
}

\begin{document}

\maketitle

\begin{abstract}
We consider the following problem:  given that a finite automaton $M$ of
$N$ states accepts at least one $k$-power-free (resp., overlap-free) word,
what is the length of the shortest such word accepted?  We give upper
and lower bounds which, unfortunately, are widely separated.
\end{abstract}

\section{Introduction}

Let $L$ be an interesting language, such as the language of primitive words,
or the language of non-palindromes.  We are interested in the following
kind of question: \emph{given that an automaton $M$ of $N$ states accepts
a member of $L$, what is a good bound
on the length $\ell(N)$ of the
shortest word accepted?}

For example,
Ito et al.\ \cite{ito} proved that if $L$ is the language of
primitive words, then $\ell(N) \leq 3N-3$.  
Horv\'ath et al.\ \cite{horvath} proved that 
if $L$ is the language of non-palindromes, then $\ell(N) \leq 3N$.  
For additional results along these lines, see
\cite{anderson}.

For an integer $k \geq 2$, a {\it $k$-power} is a nonempty word of the
form $x^k$.  A word is {\it $k$-power-free} if it has no $k$-powers as
factors.  A word of the form $axaxa$, where $a$ is a single letter, and
$x$ is a (possibly empty) word, is called an \emph{overlap}.  A word is
{\it overlap-free} if it has no factor that is an overlap.

In this paper we address two open questions left unanswered in 
\cite{anderson}, corresponding to the case
where $L$ is the language of $k$-power-free (resp., overlap-free) words.
For these words and a large enough alphabet
we give a class of DFAs of $N$ states for which the shortest 
$k$-power (resp., overlap) is of length $N^{{1\over 4} (\log N) + O(1)}$.
For overlaps over a binary alphabet we give an upper bound of $2^{O(N^{4N})}$.

\section{Notation}
For a finite alphabet $\Sigma$, let $\Sigma^*$ denote the set of finite
words over $\Sigma$. Let $w = a_0a_1\cdots a_{n-1} \in \Sigma^*$ be a word. Let
$w[i] = a_i$, and let $w[i..j] = a_i\cdots a_j$.  By convention
we have $w[i] =
\epsilon$ for $i < 0$ or $i \geq n$, and $w[i..j] = \epsilon$ for $i > j$.
A prefix $p$ of $w$ is a \emph{period} of $w$ if $w[i+r] =
w[i]$ for $0 \leq i < |w|-r$, where $r = |p|$.

For words $x,y$, let $x \preceq y$
denote that $x$ is a factor of $y$. A factor $x$ of $y$ is \emph{proper}
if $x \neq y$. Let $x \preceq_p y$ (resp.,
$x \preceq_s y$) denote that $x$ is a prefix (resp., suffix) of $y$. Let $x \prec_p y$ (resp., $x \prec_s y$) denote that $x$ is a prefix (resp., suffix) of $y$ and $x \neq y$.

A word is {\it primitive} if it is not a $k$-power for any $k \geq 2$.
Two words $x, y$ are {\it conjugate} if one is a cyclic shift of the
other; that is, if there exist words $u, v$ such
that $x = uv$ and $y = vu$. One simple observation is that all conjugates of a $k$-power are $k$-powers. 

Let $h: \Sigma^* \rightarrow \Sigma^*$ be a morphism, and suppose $h(a)
= ax$ for some letter $a$.  The {\it fixed point} of $h$, starting with
$a\in \Sigma$, is denoted by $h^\omega(a) = a\, x \, h(x) \, h^2(x) \,
\cdots$.  We say that a morphism $h$
is $k$-power-free (resp., overlap-free) if 
$h(w)$ is $k$-power-free (resp., overlap-free) if $w$ is.

Let $\Sigma_m = \{0,1,\ldots,m-1\}$.
Define the morphism $\mu:\Sigma_2^*\rightarrow \Sigma_2^*$ as follows
\begin{align*}
&\mu(0) = 01\\
&\mu(1) = 10.
\end{align*}
We call $\mathbf{t} = \mu^\omega(0)$ the {\it Thue-Morse word}. It is
easy to see that
$$\mu(\mathbf{t}[0..n-1]) = \mathbf{t}[0..2n-1] \text{
for } n\geq 0.$$ 
From classical
results of Thue \cite{thue06,thue12}, we know that the morphism $\mu$
is overlap-free. From \cite{brandenburg83}, we know that that
$\mu(x)$ is $k$-power free for each $k > 2$.

For a $\dfa D = (Q,\Sigma,\delta,q_0,F)$ the set of states, input
alphabet, transition function, set of final states, and initial state
are denoted by $Q,\Sigma,\delta,F,$ and $q_0$, respectively. Let $L(D)$
denote the language accepted by $D$. As usual, we have
$\delta(q,wa)=\delta(\delta(q,w),a)$ for a word $w$.

We state the following basic result without proof. 
\begin{proposition}\label{prop:basicTheorem}
Let $D = (Q,\Sigma,\delta,q_0,F)$ be a (deterministic or
nondeterministic) finite automaton. If $L(D)\neq \emptyset$, then $D$
accepts at least one word of length smaller than $|Q|$.
\end{proposition}

\section{Lower bound}

In this section, we construct an infinite family of $\dfa$s such that
the shortest $k$-power-free word accepted is rather long, as a function
of the number of states.  Up to now only a linear bound was known.

For a word $w$ of length $n$ and $i \geq 1$, let $$\cyc_i(w) = w[i..n-1]\,w[0.. i-2]$$ denote $w$'s $i$th cyclic shift to the left, followed by removing the last symbol. Also define $$\cyc_0(w)=w[0..n-2].$$ For example, we have
\begin{align*}\tt
&\cyc_2 ({\tt recompute}) = {\tt computer},\\
&\cyc_4 ({\tt richly}) = {\tt lyric}.
\end{align*}

 We call each $\cyc_i(w)$ a {\it partial conjugate} of $w$, which is not a reflexive, symmetric, or transitive relation.

A word $w$ is a {\it simple $k$-power} if it is a $k$-power and it contains no $k$-power as a proper factor.

We start with a few lemmas. 

\begin{lemma}\label{Lemma:distinctConjugates}
	Let $w = p^k$ be a simple $k$-power. Then the word $p$ has $|p|$ distinct conjugates.
\end{lemma}
\begin{proof}
	By contradiction. If $p^k$ is a simple $k$-power, then $p$ is a primitive word. Suppose that $p = uv = xy$ such that $x \prec_p u$ and $yx = vu$. Without loss of generality, we can assume that $xv \neq \epsilon$. Then there exists a word $t \neq \epsilon$ such that $u = xt$ and $y = tv$. From $vu=yx$ we get $$vxt = tvx.$$ Using the second theorem of Lyndon and Sch\"{u}tzenberger 
\cite{lyndon}, we get that there exists $z \neq \epsilon$ such that 
	\begin{equation*}
		vx = z^i
	\end{equation*}
	\begin{equation*}
		t = z^j
	\end{equation*}
	for some positive integers $i,j$. So $yx = z^{i+j}$, and hence $p=xy$ is not primitive, a contradiction.
	\qed
\end{proof}

\begin{lemma}\label{Lemma:distinctPartialConjugates}
	Let $w$ be a simple $k$-power of length $n$. Then we have 
	\begin{equation}\label{equivPartialConjugates}
		\cyc_i(w) = \cyc_j(w) \text{ iff }  i \equiv \modd{j} {{n \over k}}.
	\end{equation}
\end{lemma}
\begin{proof}
	Let $w = p^k$. If $i \equiv i'$ (mod $\frac{n}{k}$) and $i' < \frac{n}{k}$, then
	\begin{equation*}
		\cyc_i(w)=(p[i'..\frac{n}{k}-1]\,p[0..i'-1])^{k-1}\cyc_{i'}(p).
	\end{equation*} 
	
	Similarly,
if $j \equiv j'$ (mod $\frac{n}{k}$) and
$j' < \frac{n}{k}$, then
$$\cyc_j(w) = (p[j'..\frac{n}{k}-1]\,p[0..j'-1])^{k-1}\cyc_{j'}(p).$$
	
	If $i' = j'$, then clearly $\cyc_i(w)=\cyc_j(w)$. If $i' \neq j'$, we get that $$p[i'..\frac{n}{k}-1]\,p[0..i'-1]\neq p[j'..\frac{n}{k}-1]\,p[0..j'-1]$$ using Lemma~\ref{Lemma:distinctConjugates}, and hence $\cyc_i(w) \neq \cyc_j(w)$.\\
	\qed
\end{proof}

\begin{lemma}\label{Lemma:conjugatesOfSimplePowers}
	All conjugates of a simple $k$-power are simple $k$-powers.
\end{lemma}
\begin{proof}
	By contradiction. Let $w=p^k$ be a simple $k$-power, and let $z \neq w$ be a conjugate of $w$. Clearly $z$ is a $k$-power. Suppose $z$ contains $q^k$ and $z \neq q^k$. Thus $|q| < |p|$. Since $w$ is simple $q^k \npreceq w=p^k$. The word $x=p^{k+1}$ contains $z$ as a factor. So $x = uq^kv$, for some words $u,v \preceq p$. 
	\begin{figure}[!h]
	\begin{tikzpicture}		
	\begin{scope}[xscale = 1.75,yscale=.5]
		\draw (0,0)  rectangle (1,1)
	 	      (1,0)  rectangle (2,1)
		      (2,0)  rectangle (3,1)
		      (3,0)  rectangle (4,1);
		\path (4.1,.3) node[right]{$\cdots$}rectangle(4.5,.5);      
		\draw (4.5,0)  rectangle (5.5,1)
		      (5.5,0)  rectangle (6.5,1);
		      
		\path (0,.4) node [left]{$x=$} rectangle (0,.4);
		\path (0,1.4) node [left]{$uq^kv=$} rectangle (0,1.4);
		
		\path (.35,.5) node [above,right]{$p$} rectangle (.65,.5)
	 	      (1.35,.5) node [right]{$p$} rectangle (1.65,.5)
		      (2.35,.5) node [right]{$p$} rectangle (2.65,.5)
		      (3.35,.5) node [right]{$p$} rectangle (3.65,.5)
		      (4.85,.5) node [right]{$p$} rectangle (5.15,.5)
		      (5.85,.5) node [right]{$p$} rectangle (6.15,.5);
		
		\draw (0,1)  rectangle (.9,2)
		 	  (.9,1)  rectangle (1.8,2)
			  (1.8,1)  rectangle (2.7,2)
			  (2.7,1)  rectangle (3.6,2);
		\path (4.1,1.4)
		node[right]{$\cdots$}rectangle(4.5,1.5);      	  
		\draw (4.8,1)  rectangle (5.7,2)
			  (5.7,1)  rectangle (6.5,2);
			  
	 	\path (.30,1.5) node [above,right]{$u$} rectangle (.5,1.5)
	 	      (1.2,1.5) node [right]{$q$} rectangle (1.6,1.5)
		      (2.1,1.5) node [right]{$q$} rectangle (2.5,1.5)
		      (3,1.5) node [right]{$q$} rectangle (3.5,1.5)
		      (5.1,1.5) node [right]{$q$} rectangle (5.5,1.5)
		      (6,1.5) node [right]{$v$} rectangle (6.5,1.5);
		
		\draw[->] (.95,2.5) node[above]{$|u|$} -- (.95,2);
		\draw[->] (1.85,2.5) node[above]{$|pu|-e$} -- (1.85,2);
		\draw[->] (2.75,2.5) node[above]{$|p^2u|-2e$} -- (2.75,2);
		\draw[->] (4.85,2.5) node[above]{$|p^{k-1}u|-(k-1)e$} -- (4.85,2);
	\end{scope}
	\end{tikzpicture}
	\caption{starting positions of the occurrences of $q$ inside $x$}
	\label{positionsOfQinsideW}
	\end{figure}
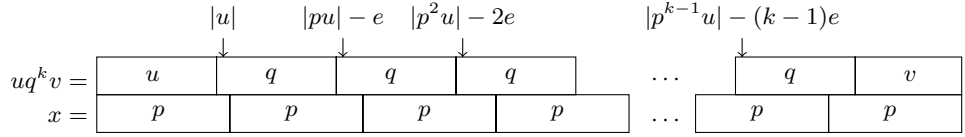\\
	 Note that $u$ and $v$ are nonempty and not equal to $p$ since $q^k \npreceq p^k$. Letting $e := |p|-|q|$, and considering
the starting positions of the occurrences of $q$ in $x$ 
(see Fig.~\ref{positionsOfQinsideW}), we can write 
	\begin{equation*}\label{equationInW}
	x[|p^iu|-ie..|p^iu|-(i-1)e-1]=x[|p^ju|-je..|p^ju|-(j-1)e-1]
	\end{equation*}
	for every $0 \leq i,j < k$. Since $p$ is a period of $x$, we can write
	\begin{equation*}
	x[|u|-ie..|u|-(i-1)e-1]=x[|u|-je..|u|-(j-1)e-1]
	\end{equation*}
	which means $x[u-(k-1)e..u+e-1] \preceq w$ is a $k$-power. Therefore $w$ contains a $k$-power other than itself, a contradiction. 
	\qed
\end{proof}

\begin{corollary}\label{corollary}
Partial conjugates of simple $k$-powers are $k$-power-free. 
\end{corollary} 

The next lemma shows that there are infinitely many simple $k$-powers
over a binary alphabet for $k>2$. We also show that there are
infinitely many simple squares over a ternary alphabet, using a result
of Currie \cite{currie02}. 

\begin{lemma}\label{simpleCubeExist} $  $
	\begin{itemize}
		\item[(i)] Let $p = \mathbf{t}[0.. 2^n - 1]$ where $n
		\geq 0$. For every $k > 2$, the word $p^k$ is a simple
		$k$-power.  \item[(ii)]\label{simpleCubeExistPart2} There are infinitely many
		simple squares over a ternary alphabet.
\end{itemize}

\end{lemma} \begin{proof} $  $ \begin{itemize}
	\item[(i)] By induction on $n$. For $n = 0$ we have $p^k = 0^k$
	which is a simple $k$-power. Suppose $n > 0$. To get a
	contradiction, suppose that there exist words $u,v,x$ with $uv
	\neq \epsilon$ and $x \neq \epsilon$ such that $p^k = ux^kv$.
	Note that $|x| < |p|$, so $|uv| \geq k$. Without loss of
	generality, we can assume that $|v| \geq
	\lceil\frac{k}{2}\rceil\geq 2$. Let $q = \mathbf{t}[0.. 2^{n-1}
	- 1]$. We know that $$p^k =\mu(q^k).$$ We can write $$w = ux^k
	\preceq_p \mu(q^{k-1}q[0..|q|-2]).$$ Since $\mu$ is
	$k$-power-free, the word $q^{k-1}q[0..|q|-2]$ contains a
	$k$-power. Hence $q^k$ contains at least two $k$-powers, a
	contradiction.    \\ \item[(ii)] Currie \cite{currie02} proved
	that over a ternary alphabet, for every $n \geq 18$, there is a
	word $p$ of length $n$ such that all its conjugates are
	squarefree. Such squarefree words are called \emph{circularly
	squarefree words}. 

	\medskip
	
	We claim that for every circularly
	squarefree word $p$, the word $p^2$ is a simple square.  To get 
	a contradiction, let $q^2$ be the smallest square in $p^2$. So
	there exist words $u,y$ with $uy \neq \epsilon$ such that $p^2
	= uq^2y$. We have $|q^2| > |p| $ since $p$ is circularly
	squarefree. Therefore, if we let $p=uv=xy$, then $|x|>|u|$ and
	$|v|>|y|$. So there exists $t$ such that $x=ut$ and $v=ty$.  We can assume $|t|<|q|$, since otherwise $|t|=|q|$ and $|uy|=0$, a contradiction. Now since $q^2 = vx=tyut$, we get that $q$ begins and ends with $t$, which means $t^2 \prec q^2$. Therefore $p^2$ has a smaller square than $q^2$, a contradiction.
\end{itemize}
	\qed \end{proof}

Next we show how to construct arbitrarily long simple $k$-powers from
smaller ones. Fix $k = 2\ (\text{resp., }k\geq 3)$ and $m = 3\ (\text{resp., }m=2)$. Let $w_1\in\Sigma_m^*$ be a simple $k$-power. Using the previous lemma, there are infinitely many choices for $w_1$. Let $w_1$ be of length $n$. Define $w_{i+1} \in\Sigma_{m+i}^*$ for $i \geq 1$ recursively by
\begin{align}\label{eq:definitionOfwi}
	w_{i+1} = \cyc_0(w_{i})a_i\cyc_{n^{i-1}}(w_{i})a_i\cyc_{2n^{i-1}}(w_{i})a_i\cdots\cyc_{(n-1)n^{i-1}}(w_{i})a_i
\end{align}
where $a_i = m+i-1$ and $w_{0} = 0$.  The next lemma states that $w_i$, for $i\geq 1$, is a simple
$k$-power. Therefore,
using Corollary \ref{corollary}, each word
$\cyc_0(w_i)$ is $k$-power-free.
For $i\geq 1$, it is easy to see that 
\begin{equation}\label{eq:length}
|w_i|=n|w_{i-1}|=n^{i}.
\end{equation}

\begin{lemma}\label{MakeLongerWords}
For every $i \geq 1$, the word $w_i$ is a simple $k$-power.
\end{lemma}
\begin{proof}
By induction on $i$. The word $w_1$ is a simple $k$-power.
Now suppose that $w_{i}$ is a simple $k$-power for some $i \geq 1$.
Using Lemma~\ref{Lemma:distinctPartialConjugates}, we have $\cyc_{jn^{i-1}}(w_i)=\cyc_{(j+\frac{n}{k})n^{i-1}}(w_i)$,
since $\frac{|w_i|}{k} = \frac{n^i}{k}$.

We now claim that $w_{i+1}$ is a $k$-power and
$$w_{i+1} = (\cyc_0(w_{i})a_i\cyc_{n^{i-1}}(w_{i})a_i\cyc_{2n^{i-1}}(w_{i})a_i\cdots\cyc_{(\frac{n}{k}-1)n^{i-1}}(w_{i})a_i)^k.$$
To see this,
suppose that $w_{i+1}$ contains a $k$-power $y^k$ such that $w_{i+1} \neq y^k$.
	
If $y$ contains more than one occurrence of $a_i$, then $y =
ua_i\cyc_j(w_i)a_iv$ for some words $u,v$ and an integer $j$. Since $y^2 =
ua_i\cyc_j(w_i)a_ivua_i\cyc_j(w_i)a_iv \preceq w_{i+1}$, using
\eqref{eq:definitionOfwi} and
Lemma~\ref{Lemma:distinctPartialConjugates}, we get
$$|y|=\left|\cyc_j(w_i)a_ivua_i\right|\geq
\frac{n}{k}n^i=\frac{|w_{i+1}|}{k},$$ and hence $y^k=w_{i+1}$, a
contradiction.

	If $y$ contains just one $a_i$, then $y = ua_iv$ for some words
$u,v$ which contain no $a_i$.  So $y^k = u(avu)^{k-1}av$ for $a =
a_i$.  Therefore $vu$ is a partial conjugate of $w_i$. However the
distance between two equal partial conjugates of $w_i$ in $w_{i+1}$ is
longer than just one letter, using \eqref{eq:definitionOfwi} and
Lemma~\ref{Lemma:distinctPartialConjugates}.

	Finally, if $y$ contains no $a_i$, then a partial conjugate of
	$w_i$ contains a $k$-power, which is impossible due to
	Lemma~\ref{Lemma:conjugatesOfSimplePowers}.  \qed \end{proof}

To make our formulas easier to read, we define $a_0 = w_1[n-1]$.

\begin{theorem}\label{thm}
	For $i \geq 1$, there is a $\dfa D_i$ with $2^{i-1}(n-1)+2$ states such that $\cyc_0(w_i)$ is the shortest $k$-power-free word in $L(D_i)$.
\end{theorem}

\begin{proof}
	Define $D_1 = (Q_1,\Sigma_{a_1},\delta_1,q_{1,0},F_1)$ where
	\begin{align*}
		&Q_{1} := \{q_{1,0},q_{1,1},q_{1,2},\ldots,q_{1,n-1},q_d\},\\
		&F_{1}:=\{q_{1,n-1}\},\\ 
		&\delta_1(q_{1,j},w[j]):=q_{1,j+1}\text{ for } 0\leq j < n-1,
	\end{align*}
	and the rest of the transitions go to the dead state $q_{d}.$  Clearly we have $|Q_1|=n+1$ and $L(D_1) = \{\cyc_0(w_1)\}$.
	
	We define $D_{i}=(Q_{i},\Sigma_{a_i},\delta_{i},q_{1,0},F_i)$ for $i \geq 2$ recursively. 
	For the rest of the proof $s$ and $t$ denote (possibly empty)
sequences of integers and $j$ denotes
a single integer (a sequence of length $1$).
	We use integer sequences as subscripts of states in $Q_{i}$. For example, $q_{1,0}$, $q_{s,j}$, and $q_{s,2,t}$ might denote states of $D_i$. For $i\geq 1$, define 
	\begin{align}
		&Q_{i+1} := Q_i \cup \{q_{i+1,t}:q_t\in (Q_i-F_i)-\{q_d\}\},\label{recursiveDefinition:states}\\
		&F_{i+1}:=\{q_{i+1,i,t}: \delta_{i}(q_{i,t},c)=q_{1,n-1} \text{ for some } c \in \Sigma_{a_i}\},\label{recursiveDefinition:finalState}\\ 
		&\text{if } q_t \in Q_i \text{ and } c \in \Sigma_{a_i},  \text{ then }\delta_{i+1}(q_t,c):=\delta_i(q_t,c)\label{recursiveDefinition:oldTransitions}\\
		&\text{if } q_t,q_s \in (Q_i-F_i)-\{q_d\}, c \in \Sigma_{a_i}, \text{ and }  \delta_i(q_t,c)=q_s,\nonumber\\ &\hspace{2cm}\text{ then } \delta_{i+1}(q_{i+1,t},c):=q_{i+1,s}\label{recursiveDefinition:oldTransitionsRenewed}\\
		&\text{if } q_t \in F_i, \text{ then } \delta_{i+1}(q_t,a_i):=q_{1,1} \text{ and } \delta_{i+1}(q_t,a_{i-1}):=q_{i+1,1,0}\label{recursiveDefinition:transitionType1}\\
		&\text{if } i > 1, q_{i+1,t} \notin F_{i+1}, \text{ and } \delta_i(q_{t},a_{i-1})=q_{1,j},\nonumber\\ &\hspace{2cm}\text{ then } \delta_{i+1}(q_{i+1,t},a_i):=q_{1,j+1}\label{recursiveDefinition:transitionType2}
	\end{align}
	and finally for the special case of $i=1$,
	\begin{equation}
		\delta_2(q_{2,1,j},a_1):=q_{1,j+2} \text{ for } 0\leq j < n-2.\label{recursiveDefinition:transitionType2iequals1}
	\end{equation}
	The rest of the transitions, not indicated 
in \eqref{recursiveDefinition:oldTransitions}--\eqref{recursiveDefinition:transitionType2iequals1}, go to the dead state $q_{d}$. Fig.~\ref{fig:transitionDiagramB} depicts $D_2$ and $D_3$. Using \eqref{recursiveDefinition:states}, we have $|Q_{i+1}|= 2|Q_i|-2=2^i(n-1)+2$ by a simple induction.
	
	An easy induction on $i$ proves that $|F_i|=1$. So let $f_i$ be the appropriate integer sequence for which $F_i = \{q_{f_i}\}$. Using 
\eqref{recursiveDefinition:oldTransitions}--\eqref{recursiveDefinition:transitionType2iequals1}, we get that for every $1\leq j < n$, there exists exactly one state $q_t \in Q_i$ for which $\delta_i(q_t,a_{i-1})= q_{1,j}$. 
		
	By induction on $i$, we prove that for $i\geq 2$ if $\delta_i(q_t,a_{i-1})=q_{1,j}$, then
	\begin{align}
		&x_1=\cyc_{(j-1)n^{i-2}}(w_{i-1}),\label{eq:firstWord}\\
		&x_2=w_i[0..jn^{i-1}-2],\label{eq:secondWord}\\
		&x_3=w_i[(j-1)n^{i-1}..n^i-2].\label{eq:thirdWord}
	\end{align} are the shortest $k$-power-free words for which 
	\begin{align} 
	&\delta_i(q_{1,j-1},x_1)=q_t,\label{eq:firstProperty}\\
	&\delta_i(q_{1,0},x_2)=q_t,\label{eq:secondProperty}\\
	&\delta_i(q_{1,j-1},x_3) = q_{f_i}.\label{eq:thirdProperty}
	\end{align}
	
	In particular, from \eqref{eq:thirdWord} and \eqref{eq:thirdProperty}, for $j=1$, we get that $\cyc_0(w_i)$ is the shortest $k$-power-free word in $L(D_i)$.
	
	The fact that our choices of $x_1,x_2,$ and $x_3$ are $k$-power-free follows from the fact that proper factors of simple $k$-powers are $k$-power-free.
	For $i=2$ the proofs of \eqref{eq:firstProperty}--\eqref{eq:thirdProperty} are easy and left to the readers.
	
	Suppose that \eqref{eq:firstProperty}--\eqref{eq:thirdProperty} hold for some $i \geq 2$. Let us prove \eqref{eq:firstProperty}--\eqref{eq:thirdProperty} for $i+1$. Suppose that 
	\begin{equation}\label{eq:assumption}
		\delta_{i+1}(q_t,a_{i})=q_{1,j}.
	\end{equation} 
	First we prove that the shortest $k$-power-free word $x$ for which 
		$$\delta_{i+1}(q_{1,j-1},x)=q_t,$$
	is
		$x=\cyc_{(j-1)n^{i-1}}(w_i).$
	
	If $q_t \in Q_{i}$, from \eqref{recursiveDefinition:transitionType1} and \eqref{eq:assumption}, we have 
	\begin{align*}
		&q_t = q_{f_i}, \text{ and}\\
		&\delta_{i+1}(q_t,a_i)=q_{1,1}.
	\end{align*}
	By induction hypothesis, the $\cyc_0(w_{i})$ is the shortest $k$-power-free word in $L(D_i)$. In other words, we have $\delta_{i}(q_{1,0},\cyc_0(w_i))=q_{f_i}=q_t$, which can be rewritten using \eqref{recursiveDefinition:oldTransitions} as $\delta_{i+1}(q_{1,0},\cyc_0(w_i))=q_t$.
	
	Now suppose $q_t \notin Q_i$. Then by \eqref{recursiveDefinition:transitionType2} and \eqref{eq:assumption}, we get that there exists $t'$ such that
	\begin{align*}
		&t = i+1,t';\\
		&\delta_i(q_{t'},a_{i-1})=q_{1,j-1}.
	\end{align*}
	From the induction hypothesis, i.e., \eqref{eq:secondProperty} and \eqref{eq:thirdProperty}, we can write
	\begin{align}
	&\delta_i(q_{1,0},w_i[0..(j-1)n^{i-1}-2])=q_{t'},\label{eq:1revisited}\\
	&\delta_i(q_{1,j-1},w_i[(j-1)n^{i-1}..n^i-2])=q_{f_i}.\label{eq:2revisited}
	\end{align}
	In addition $w_i[0..(j-1)n^{i-1}-2]$ and $w_i[(j-1)n^{i-1}..n^i-2]$ are the shortest $k$-power-free transitions from $q_{1,0}$ to $q_{t'}$ and from $q_{1,j-1}$ to $q_{f_i}$ respectively. 
Using \eqref{recursiveDefinition:oldTransitions}, we can rewrite \eqref{eq:1revisited} and \eqref{eq:2revisited} for $\delta_{i+1}$ as follows:
	\begin{align}
	&\delta_{i+1}(q_{1,0},w_i[0..(j-1)n^{i-1}-2])=q_{t'},\label{eq:3revisited}\\
	&\delta_{i+1}(q_{1,j-1},w_i[(j-1)n^{i-1}..n^i-2])=q_{f_i}.\label{eq:4revisited}
	\end{align}
	Note that from \eqref{recursiveDefinition:oldTransitionsRenewed} and \eqref{eq:3revisited}, we get
	\begin{equation}
	\delta_{i+1}(q_{i+1,1,0},w_i[0..(j-1)n^{i-1}-2])=q_{i+1,t'}=q_t.\label{eq:5revisited}
	\end{equation}
	We also have $\delta_{i+1}(q_
	{f_i},a_i)=q_{i+1,1,0}$, using \eqref{recursiveDefinition:transitionType1}. So together with \eqref{eq:4revisited} and \eqref{eq:5revisited}, we get
	$$\delta_{i+1}(q_{1,j-1},\cyc_{(j-1)n^{i-1}}(w_{i}))=q_t$$ and $\cyc_{(j-1)n^{i-1}}(w_{i})$ is the shortest $k$-power-free transition from $q_{1,j-1}$ to $q_t$. 
	
	The proofs of \eqref{eq:secondProperty} and \eqref{eq:thirdProperty} are similar. 
	\qed
\end{proof}

In what follows, all logarithms are to the base 2.

\begin{corollary}\label{corollary:main}
For infinitely many $N$, there exists a $\dfa$ with $N$ states such
that the shortest $k$-power-free word accepted is of length
$N^{{1 \over 4}{\log N} + O(1)}$.
\end{corollary}
\begin{proof}
Let $i = \lfloor \log n \rfloor$ in Theorem~\ref{thm}.
Then $D = D_i$ has 
$$N = 2^{\lfloor \log n\rfloor-1}(n-1)+2 = \Omega(n^2)$$
states. 
In addition, the shortest $k$-power-free word in $L(D)$ is
$\cyc_0\left(w_{\lfloor \log n \rfloor}\right)$.
Now, using \eqref{eq:length} 
we can write 
$$\left|\cyc_0(w_{\lfloor \log n \rfloor})\right|
=n^{\lfloor \log n \rfloor}-1.$$
Suppose $2^t \leq n < 2^{t+1} -1$, so that $t = \lfloor \log n \rfloor$
and 
Then $\log N = 2t + O(1)$, so
${1 \over 4}(\log N)^2 = t^2 + O(t)$.
On the other hand $\log |w| = \lfloor \log n \rfloor (\log n) = t (t+ O(1))
= t^2 + O(t)$.    Now $2^{O(t)} = n^{O(1)} = N^{O(1)}$, and
the result follows.
\qed
\end{proof}

\begin{remark}
The same bound holds for overlap-free words.  To do so, we
define a {\it simple overlap} as a word of the form $axaxa$ where
$axax$ is a simple square.  In our construction of the DFAs, we use
complete conjugates of $(ax)^2$ instead of partial conjugates.
\end{remark}

\begin{remark}
The $D_i$ in Theorem~\ref{thm} are defined over the growing
alphabet $\Sigma_{m+i-1}$.
However, we can fix the alphabet to be $\Sigma_{m+1}$.
For this purpose, we introduce $w'_i$ which is quite similar to $w_i$:
\begin{align*}
	&w'_1 = w_1,\\
	&w'_{i+1} = \cyc_0(w'_{i})b_i\cyc_{n^{i-1}}(w'_{i})b_i\cyc_{2n^{i-1}}(w'_{i})b_i\cdots\cyc_{(n-1)n^{i-1}}(w'_{i})b_i,
\end{align*}
where $b_i = mc_im$ such 
that $c_i$ is (any of) the shortest nonempty $k$-power-free word over
$\Sigma_{m}$ not equal to $c_1,\ldots, c_{i-1}$. Clearly we have $|b_i|
\leq |b_{i-1}|+1 = O(i)$, and hence $w'_i = \Theta(n^i)$.

One can then prove Lemma~\ref{MakeLongerWords} and Theorem~\ref{thm} for $w'_i$ with minor modifications of the argument above. In particular, we construct $\dfa D'_i$ that accepts $\cyc_0(w'_i)$ as the shortest $k$-power-free word accepted, 
and a $D'_i$ that is quite similar to $D_i$. In particular, they have asymptotically the same number of states.
\end{remark}

\begin{figure}[H]				            				
\subfloat[][transition diagram of $D_2$]{\label{fig:transitionDiagramA}
\begin{tikzpicture}[->,>=stealth',shorten >=.5pt,auto,node distance=1.5cm, semithick] \tikzstyle{state} = [draw, semithick, fill=white, circle, minimum height=3.0em, minimum width=3.0em, node distance=5.0em]\tikzstyle{every state}=[draw=black]
\node[initial,state]   (A)  {$q_{1,0}$};
\node[state]           (B) [below of=A]  {$q_{1,1}$};
\node[state]           (C) [below of=B]  {$q_{1,2}$};
\node[state]           (D) [below of=C]  {$q_{1,3}$};
\node[draw=none, fill = none] (E) [below of=D]  {$\vdots$};
\node[state]           (F) [below of=E]  {$q_{1,n-2}$};
\node[state]           (G) [below of=F]  {$q_{1,n-1}$};
\node[state]           (H) [below of=G]  {$q_{2,1,0}$};
\node[state]           (I) [below of=H]  {$q_{2,1,1}$};
\node[draw=none, fill = none] (J) [below of=I]  {$\vdots$};
\node[state]           (K) [below of=J]  {$q_{2,1,n-4}$};
\node[state]           (L) [below of=K]  {$q_{2,1,n-3}$};
\node[state,accepting] (M) [below of=L]  {$q_{2,1,n-2}$};

\path (A) edge node {\color{blue}\tiny $w_1[0]$} (B);
\path (B) edge node {\color{blue}\tiny $w_1[1]$} (C);
\path (C) edge node {\color{blue}\tiny $w_1[2]$} (D);
\path (D) edge node {\color{blue}\tiny $w_1[3]$} (E);
\path (E) edge node {\color{blue}\tiny $w_1[n-3]$} (F);
\path (F) edge node {\color{blue}\tiny $w_1[n-2]$} (G);
\path (G) edge node {\color{blue}\tiny $a_0=w_1[n-1]$} (H);
\path (H) edge node {\color{blue}\tiny $w_1[0]$} (I);
\path (I) edge node {\color{blue}\tiny $w_1[1]$} (J);
\path (J) edge node {\color{blue}\tiny $w_1[n-5]$} (K);
\path (K) edge node {\color{blue}\tiny $w_1[n-4]$} (L);
\path (L) edge node {\color{blue}\tiny $w_1[n-3]$} (M);

\path (G) edge [bend left] node {\color{red}\tiny $a_1$} (B);
\path (H) edge [bend left] node {\color{red}\tiny $a_1$} (C);
\path (I) edge [bend left] node {\color{red}\tiny $a_1$} (D);
\path (K) edge [bend left] node {\color{red}\tiny $a_1$} (F);
\path (L) edge [bend left] node {\color{red}\tiny $a_1$} (G);
\end{tikzpicture}
}\subfloat[][transition diagram of $D_3$]{\label{fig:transitionDiagramB}
\begin{tikzpicture}[->,>=stealth',shorten >=.5pt,auto,node distance=1.5cm, semithick] \tikzstyle{state} = [draw, semithick, fill=white, circle, minimum height=3.0em, minimum width=3.0em, node distance=5.0em]\tikzstyle{every state}=[draw=black]

\node[initial,state]   (A)  {$q_{1,0}$};
\node[state]           (B) [below of=A]  {$q_{1,1}$};
\node[state]           (C) [below of=B]  {$q_{1,2}$};
\node[state]           (D) [below of=C]  {$q_{1,3}$};
\node[draw=none, fill = none] (E) [below of=D]  {$\vdots$};
\node[state]           (F) [below of=E]  {$q_{1,n-2}$};
\node[state]           (G) [below of=F]  {$q_{1,n-1}$};
\node[state]           (H) [below of=G]  {$q_{2,1,0}$};
\node[state]           (I) [below of=H]  {$q_{2,1,1}$};
\node[draw=none, fill = none] (J) [below of=I]  {$\vdots$};
\node[state]           (K) [below of=J]  {$q_{2,1,n-4}$};
\node[state]           (L) [below of=K]  {$q_{2,1,n-3}$};
\node[state] (M) [below of=L]  {$q_{2,1,n-2}$};

\path (A) edge node {\color{blue}\tiny $w_1[0]$} (B);
\path (B) edge node {\color{blue}\tiny $w_1[1]$} (C);
\path (C) edge node {\color{blue}\tiny $w_1[2]$} (D);
\path (D) edge node {\color{blue}\tiny $w_1[3]$} (E);
\path (E) edge node {\color{blue}\tiny $w_1[n-3]$} (F);
\path (F) edge node {\color{blue}\tiny $w_1[n-2]$} (G);
\path (G) edge node {\color{blue}\tiny $a_0=w_1[n-1]$} (H);
\path (H) edge node {\color{blue}\tiny $w_1[0]$} (I);
\path (I) edge node {\color{blue}\tiny $w_1[1]$} (J);
\path (J) edge node {\color{blue}\tiny $w_1[n-5]$} (K);
\path (K) edge node {\color{blue}\tiny $w_1[n-4]$} (L);
\path (L) edge node {\color{blue}\tiny $w_1[n-3]$} (M);

\path (G) edge [bend left=50] node {\color{red}\tiny $a_1$} (B);
\path (H) edge [bend left=50] node {\color{red}\tiny $a_1$} (C);
\path (I) edge [bend left=50] node {\color{red}\tiny $a_1$} (D);
\path (K) edge [bend left=50] node {\color{red}\tiny $a_1$} (F);
\path (L) edge [bend left=50] node {\color{red}\tiny $a_1$} (G);

\node[state][xshift=4cm]   (A2)  {$q_{3,1,0}$};
\node[state]           (B2) [below of=A2]  {$q_{3,1,1}$};
\node[state]           (C2) [below of=B2]  {$q_{3,1,2}$};
\node[state]           (D2) [below of=C2]  {$q_{3,1,3}$};
\node[draw=none, fill = none] (E2) [below of=D2]  {$\vdots$};
\node[state]           (F2) [below of=E2]  {$q_{3,1,n-2}$};
\node[state]           (G2) [below of=F2]  {$q_{3,1,n-1}$};
\node[state]           (H2) [below of=G2]  {$q_{3,2,1,0}$};
\node[state]           (I2) [below of=H2]  {$q_{3,2,1,1}$};
\node[draw=none, fill = none] (J2) [below of=I2]  {$\vdots$};
\node[state]           (K2) [below of=J2]  {$q_{3,2,1,n-4}$};
\node[state,accepting]           (L2) [below=.25cm of K2]  {$q_{3,2,1,n-3}$};

\path (A2) edge node {\color{blue}\tiny $w_1[0]$} (B2);
\path (B2) edge node {\color{blue}\tiny $w_1[1]$} (C2);
\path (C2) edge node {\color{blue}\tiny $w_1[2]$} (D2);
\path (D2) edge node {\color{blue}\tiny $w_1[3]$} (E2);
\path (E2) edge node {\color{blue}\tiny $w_1[n-3]$} (F2);
\path (F2) edge node {\color{blue}\tiny $w_1[n-2]$} (G2);
\path (G2) edge node {\color{blue}\tiny $a_0=w_1[n-1]$} (H2);
\path (H2) edge node {\color{blue}\tiny $w_1[0]$} (I2);
\path (I2) edge node {\color{blue}\tiny $w_1[1]$} (J2);
\path (J2) edge node {\color{blue}\tiny $w_1[n-5]$} (K2);
\path (K2) edge node {\color{blue}\tiny $w_1[n-4]$} (L2);

\path (G2) edge [bend right=50] node {\color{red}\tiny $a_1$} (B2);
\path (H2) edge [bend right=50] node {\color{red}\tiny $a_1$} (C2);
\path (I2) edge [bend right=50] node {\color{red}\tiny $a_1$} (D2);
\path (K2) edge [bend right=50] node {\color{red}\tiny $a_1$} (F2);
\path (L2) edge [bend right=50] node {\color{red}\tiny $a_1$} (G2);

\path (M) edge [red,bend left=45]node {\color{green}\tiny $a_2$} (B);
\path (M) edge[red] node[near end] {\color{red}\tiny $a_1$} (A2);
\path (K2) edge[red] node[very near start] {\color{green}\tiny $a_2$} (G);
\path (H2) edge[red] node[very near start] {\color{green}\tiny $a_2$} (D);
\path (G2) edge[red] node[very near start] {\color{green}\tiny $a_2$} (C);
\path[dotted] (I2) edge[red] node[very near start] {\color{green}\tiny $a_2$} (E);
\end{tikzpicture}
}

\caption{transition diagrams}          
\label{fig:transitionDiagram}

\end{figure}
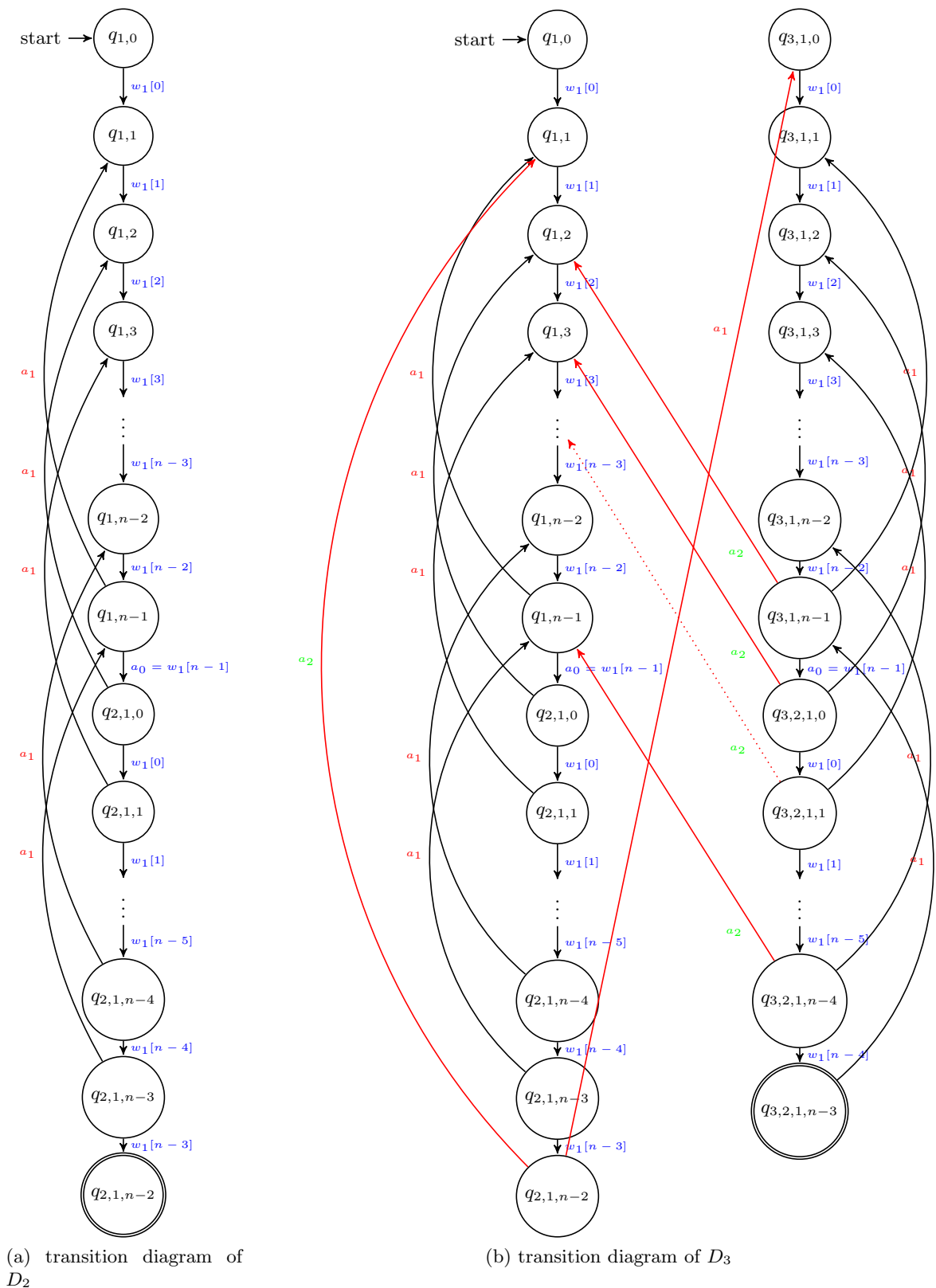

\section{Upper bound for overlap-free words}
In this section, we prove an upper bound on the length of the shortest overlap-free word accepted by a $\dfa D$ over a binary alphabet.

Let $L = L(D)$ and let $R$ be the set of overlap-free words over $\Sigma_2^*$. Carpi \cite{carpi93} defined a certain
operation $\Psi$ on binary languages, and proved that $\Psi(R)$ is regular. We prove that $\Psi(L)$ is also regular, and hence $\Psi(L) \cap \Psi(R)$ is regular. The next step is to apply Proposition~\ref{prop:basicTheorem} to get an upper bound on the length of the shortest word in $\Psi(L) \cap \Psi(R)$. This bound then gives us an upper bound on the length of the shortest overlap-free word in $L$.

Let $H = \{\epsilon,0,1,00,11\}$. Carpi defines maps 
\begin{align*}
	&\Phi_l,\Phi_r: \Sigma_{25} \rightarrow H
\end{align*}
such that for every pair $h,h' \in H$, one has 
\begin{equation*}
	h = \Phi_l(a), h' = \Phi_r(a)
\end{equation*}
for exactly one letter $a \in \Sigma_{25}$. 

For every word $w \in \Sigma_{25}^*$, define $\Phi(w)\in\Sigma_2^*$ inductively by
\begin{equation}\label{eq:PhiDefinition}
	\Phi(\epsilon) = \epsilon, \Phi(aw) = \Phi_l(a)\mu(\Phi(w))\Phi_r(a)  \hspace{20pt} (w \in \Sigma_{25}^*, a \in \Sigma_{25}).
\end{equation}
Expanding \eqref{eq:PhiDefinition} for $w = a_0a_1\cdots a_{n-1}$, we get
\begin{equation}\label{eq:expanded}
	\Phi_l(a_{0})\mu(\Phi_l(a_{1}))\cdots \mu^{n-1}(\Phi_l(a_{n-1}))\mu^{n-1}(\Phi_r(a_{n-1}))\cdots \mu(\Phi_r(a_{1}))\Phi_r(a_{0}).
\end{equation}

For $L\subseteq \Sigma_2^*$ define 
$\Psi(L) = \bigcup_{x\in L} \Phi^{-1}(x)$.
Based on the decomposition of Restivo and Salemi \cite{restivo} for finite overlap-free words, the language $\Psi(x)$ is always nonempty for an overlap-free word $x\in\Sigma_2^*$. The next theorem is due to Carpi \cite{carpi93}. 

\begin{theorem} 
$\Psi(R)$ is regular. 
\end{theorem}
Carpi constructed a $\dfa A$ with less than $400$ states that accepts $\Psi(R)$. We prove that $\Psi$ preserves regular languages.

\begin{theorem}\label{thm:preservingRegularLanguages}
Let $D=(Q,\Sigma_2,\delta,q_0,F)$ be a $\dfa$ with $N$ states, and let $L = L(D)$. Then $\Psi(L)$ is regular and is accepted by a $\dfa$ with at most $N^{4N}$ states.
\end{theorem}
\begin{proof}
Let $\iota :Q\rightarrow Q$ denote the identity function, and define $\eta_0,\eta_1:Q\rightarrow Q$ as follows
\begin{equation}\label{eq:functinal}
	\eta_i(q) = \delta(q,i) \text{ for } i=0,1.
\end{equation}
For functions $\zeta_0,\zeta_1:Q\rightarrow Q$, and a word $x=b_0b_1\cdots b_{n-1} \in \Sigma_2^*$, define $\zeta_x = \zeta_{b_{n-1}}\circ\cdots\circ \zeta_{b_1}\circ\zeta_{b_{0}}.$ Therefore we have $\zeta_{y}\circ\zeta_{x}=\zeta_{xy}$. Also by convention $\zeta_\epsilon = \iota$. So for example $x \in L(D)$ if and only if $\eta_x(q_0)\in F$.

We create $\dfa D' = (Q',\Sigma_{25},\delta',q_0',F')$ where
\begin{align*}
&Q' = \{[\kappa,\lambda,\zeta_0,\zeta_1] : \kappa,\lambda,\zeta_0,\zeta_1: Q \rightarrow Q\},\\
&\delta'([\kappa,\lambda,\zeta_0,\zeta_1],a) = \left[ \zeta_{\Phi_l(a)}\circ\kappa,\lambda\circ\zeta_{\Phi_r(a)} ,\zeta_1\circ\zeta_0,\zeta_0\circ\zeta_1\right].
\end{align*}
Also let 
\begin{align}
&q_0' = [\iota,\iota,\eta_0,\eta_1],\nonumber\\
&F' = \{[\kappa,\lambda,\zeta_0,\zeta_1]: \lambda\circ\kappa(q_0) \in F\}.\label{eq:finalState}
\end{align}

We can see that $|Q'|=N^{4N}$. We claim that $D'$ accepts $\Psi(L)$.
Indeed, on input $w$, the $\dfa D'$ simulates the behavior of $D$ on $\Phi(w)$. 

Let $w=a_0a_1\cdots a_{n-1}\in \Sigma_{25}^*$, and define 
$$\Phi_1(w)=\Phi_l(a_{a_0})\mu(\Phi_l(a_{1}))\cdots \mu^{n-1}(\Phi_l(a_{n-1})),$$
$$\Phi_2(w)=\mu^{n-1}(\Phi_r(a_{n-1}))\cdots \mu(\Phi_r(a_{1}))\Phi_r(a_{0}).$$
Using \eqref{eq:expanded}, we can write $$\Phi(w)=\Phi_1(w)\Phi_2(w).$$

We prove by induction on $n$ that 
\begin{equation}\label{eq:proofCore}
\delta'(q_0',w)=\left[\eta_{\Phi_1(w)},\eta_{\Phi_2(w)},\eta_{\mu^n(0)},\eta_{\mu^n(1)}\right].
\end{equation}

For $n=0$, we have $\Phi(w)=\Phi_1(w)=\Phi_2(w)=\epsilon$. So 
$$
	\delta'(q_0',\epsilon) = q_0'
	=[\iota,\iota,\eta_0,\eta_1] 
	=[\eta_{\Phi_1(w)},\eta_{\Phi_2(w)},\eta_{\mu^0(0)},\eta_{\mu^0(1)}].
$$


So we can assume \eqref{eq:proofCore} holds for some $n \geq 0$. Now 
suppose $w=a_0a_1\cdots a_n$ and write
\begin{align}
	&\delta'(q_0',a_0a_1\cdots a_n)\nonumber\\
	&=\delta'(\delta'(q_0',a_0a_1\cdots a_{n-1}),a_n)\nonumber\\	
	&=\delta'\left(\left[\eta_{\Phi_1(w[0..n-1])},\eta_{\Phi_2(w[0..n-1])},\eta_{\mu^{n}(0)},\eta_{\mu^{n}(1)}\right],a_n\right)\nonumber\\
	&=\left[ \eta_{\mu^n(\phi_l(a_n))}\circ\eta_{\Phi_1(w[0..n-1])}, \eta_{\Phi_2(w[0..n-1])}\circ\eta_{\mu^n(\phi_r(a_n))}, \eta_{\mu^{n}(1)}\circ\eta_{\mu^{n}(0)},\eta_{\mu^{n}(0)}\circ\eta_{\mu^{n}(1)}\right]\nonumber\\
	&=\left[\eta_{\Phi_1(w)},\eta_{\Phi_2(w)},\eta_{\mu^{n+1}(0)},\eta_{\mu^{n+1}(1)}\right],\label{eq:critical}
\end{align}
and equality \eqref{eq:critical} holds because 
\begin{align*}
	&\Phi_1(w[0..n-1])\mu^n(\phi_l(a_n))=\Phi_1(w),\\
	&\mu^n(\phi_r(a_{n}))\Phi_2(w[0..n-1])=\Phi_2(w),\\
	&\mu^n(0)\mu^n(1)=\mu^n(01)=\mu^n(\mu(0))=\mu^{n+1}(0),\text{ and similarly}\\
	&\mu^n(1)\mu^n(0)=\mu^{n+1}(1).
\end{align*}
Finally, using \eqref{eq:finalState}, we have 
\begin{align*}
w\in L(D') &\iff \delta'(q_0',w)=\left[\eta_{\Phi_1(w)},\eta_{\Phi_2(w)},\zeta_0,\zeta_1\right]\in F'\\
&\iff \eta_{\Phi_1(w)}\circ\eta_{\Phi_2(w)}(q_0) \in F\\ 
&\iff\Phi(w) = \Phi_1(w)\Phi_2(w) \in L(D).
\end{align*}
\qed
\end{proof}

\begin{theorem}
Let $D=(Q,\Sigma_2,\delta,q_0,F)$ be a $\dfa$ with $N$ states. If $D$
accepts at least one overlap-free word, then the length of the shortest
overlap-free word accepted is $2^{O(N^{4N})}.$
\end{theorem}

\begin{proof}
Let $L = L(D)$. Using Theorem~\ref{thm:preservingRegularLanguages}, there exists a $\dfa D'$ with $N^{4N}$ states that accepts the language $\Psi(L)$. 

Since $\Psi(R)$ is regular and is accepted by a $\dfa$ with at most
$400$ states, we see that $$K=\Psi(L) \cap \Psi(R)$$ is regular and is
accepted by a $\dfa$ with $O\left(N^{4N}\right)$ states.

Since $L$ accepts an overlap-free word, the language $K$ is nonempty. Using
Proposition~\ref{prop:basicTheorem}, we see that $K$ contains a word $w$ of
length $O\left(N^{4N}\right)$.

Therefore $\Phi(w)$ is an overlap-free word in $L$. By induction, one
can easily prove that $|\Phi(w)| = O\left(2^{|w|}\right)$. Hence we
have $|\Phi(w)| = 2^{O(N^{4N})}$.
\qed
\end{proof}


\begin{thebibliography}{99}

\bibitem{anderson}
T. Anderson, J. Loftus, N. Rampersad, N. Santean, and J. Shallit.
\newblock Detecting palindromes, patterns and borders in regular languages.
\newblock {\it Info. Comput.} {\bf 207} (2009), 1096-1118.

\bibitem{brandenburg83}
F.-J. Brandenburg.
\newblock Uniformly growing k-th power-free homomorphisms.
\newblock \emph{Theoret. Comput. Sci.} {\bf 23} (1983), 69--82.

\bibitem{carpi93}
A. Carpi.
\newblock Overlap-free words and finite automata.
\newblock \emph{Theoret. Comput. Sci.} {\bf 115} {1993}, 243--260. 


\bibitem{currie02}
J. Currie.
\newblock There are ternary circular square-free words of length $n$ for $n \geq 18$.
\newblock \emph{Electron. J. Comb.} {\bf 9}(1) (2002), Paper \#N10.
\newblock Available at \url{http://www.combinatorics.org/ojs/index.php/eljc/article/view/v9i1n10}.

\bibitem{harju86}
T. Harju.
\newblock On cyclically overlap-free words in binary alphabets.
\newblock In G. Rozenberg and A. Salomaa, eds., 
{\it The Book of $L$}, Springer-Verlag, 1986, pp.\ 125--130. 

\bibitem{horvath}
S. Horv\'ath, J. Karhum\"aki, and J. Kleijn.
\newblock Results concerning palindromicity.
\newblock \emph{J. Info. Process. Cybern. EIK} {\bf 23} (1987), 441--451.

\bibitem{ito}
M. Ito, M. Katsura, H. J. Shyr, and S. S. Yu.
\newblock Automata accepting primitive words.
\newblock \emph{Semigroup Forum} {\bf 37} (1988), 45--52.

\bibitem{lyndon}
R. C. Lyndon and M. P. {Sch\"utzenberger}.
\newblock The equation {$a^M = b^N c^P$} in a free group.
\newblock \emph{Michigan Math. J.} {\bf 9} (1962), 289--298.

\bibitem{restivo}
A. Restivo and S. Salemi. 
\newblock On weakly square-free words.
\newblock \emph{Bull. EATCS} {\bf 21} (1983), 49-56.

\bibitem{thue06}
A. Thue.
\newblock \"{U}ber unendliche Zeichenreihen.
\newblock \emph{Norske vid. Selsk. Skr. Mat. Nat. Kl.}
{\bf 7} (1906), 1--22.
\newblock Reprinted in T. Nagell, ed., \emph{Selected Mathematical 
Papers of Axel Thue}, Universitetsforlaget, Oslo, 1977, pp.\ 139--158.

\bibitem{thue12}
A. Thue.
\newblock \"{U}ber die gegenseitige Lage gleicher Teile gewisser Zeichen
reihen.  
\newblock \emph{Norske vid. Selsk. Skr. Mat. Nat. Kl.} {\bf 1} (1912),
1--67. 
\newblock Reprinted in T. Nagell, ed.,
\emph{Selected Mathematical Papers of Axel Thue,}
Universitetsforlaget, Oslo, 1977, pp.\ 413--478.

\end{thebibliography}
\end{document}